\documentclass[conference]{IEEEtran}

\usepackage{tikz}

\usepackage{xcolor}

\usepackage{amsmath}
\usepackage{amsthm}
\usepackage{amssymb}
\usepackage{graphicx} 
  	\graphicspath{{./figures_v2/}}
\usepackage{subfigure}
\usepackage{url}
\usepackage{booktabs} 
\usepackage{array} 
\usepackage{verbatim} 
\usepackage{caption}

\usepackage{enumitem}
\usepackage{bbm}

\usepackage[font=footnotesize]{caption}

\newtheorem{theorem}{Theorem}

\newtheorem{definition}{Definition}
\newtheorem{lemma}[theorem]{Lemma}
\newtheorem{problem}{Optimization Problem}

\usepackage{color}

\newcommand{\eq}[1]{Eq.~\eqref{#1}}
\usepackage{soul}
\setstcolor{red}

\newcommand{\oldtext}[1]{}

\newcommand{\myitem}[1]{\vspace{0.25\baselineskip}\noindent\textbf{#1}}
\newcommand{\myitemit}[1]{\vspace{0.25\baselineskip}\noindent\textit{#1}}
\newcommand{\secref}[1]{Section~\ref{#1}}
\newcommand{\probref}[1]{Optimization Problem~\ref{#1}}

\usepackage{algorithm}
\usepackage{algorithmicx}
\usepackage{algpseudocode}

\usepackage[export]{adjustbox}

\IEEEoverridecommandlockouts\IEEEpubid{\makebox[\columnwidth]{978-1-5386-4725-7/18/\$31.00 \textsuperscript{\textcopyright} 2018 IEEE \hfill} \hspace{\columnsep}\makebox[\columnwidth]{ }}
\begin{document}

\title{Show me the Cache: Optimizing Cache-Friendly Recommendations for Sequential Content Access} 

\author{
\large{
Theodoros Giannakas\textsuperscript{1}, 
Pavlos Sermpezis\textsuperscript{2},        
and Thrasyvoulos Spyropoulos\textsuperscript{1}
}\\
\normalsize
\textsuperscript{1}~EURECOM,~France, first.last@eurecom.fr\\
\textsuperscript{2}~FORTH,~Greece, sermpezis@ics.forth.gr 
}
	
\maketitle

\begin{abstract}
Caching has been successfully applied in wired networks, in the context of Content Distribution Networks (CDNs), and is quickly gaining ground for wireless systems. Storing popular content at the edge of the network (e.g. at small cells) is seen as a ``win-win'' for both the user (reduced access latency) and the operator (reduced load on the transport network and core servers). Nevertheless, the much smaller size of such edge caches, and the volatility of user preferences suggest that standard caching methods do not suffice in this context. What is more, simple popularity-based models  commonly used (e.g. IRM) are becoming outdated, as users often consume multiple contents in sequence (e.g. YouTube, Spotify), and this consumption is driven by recommendation systems. The latter presents a great opportunity to bias the recommender to minimize content access cost (e.g. maximizing cache hit rates). To this end, in this paper we first propose a Markovian model for recommendation-driven user requests. We then formulate the problem of biasing the recommendation algorithm to minimize access cost, while maintaining acceptable recommendation quality. We show that the problem is non-convex, and propose an iterative ADMM-based algorithm that outperforms existing schemes, and shows significant potential for performance improvement on real content datasets.

\end{abstract}

\section{Introduction}
\label{sec:intro}
At the dawn of the era of 5G, we witness dramatic changes in \textit{network architecture} and \textit{user traffic demand}, two main elements that shape the ecosystem of modern and future mobile networks. 
New architectures for cellular networks, comprising densification of the access network (small-cells, SCs) and integration of computing and caching capabilities in base stations (mobile edge caching \& computing, MEC), have been proposed~\cite{bastug2014living}, to cope with the recent boom in traffic demand~\cite{networking2016forecast} and the envisioned increase in devices/traffic density in the near future ($\times10k$ more traffic and $\times10-100$ more devices~\cite{5GnumbersWashingtonPost}).
At the same time, the higher data rates offered by mobile operators, lead to changes in the consuming habits of users as well. Users nowadays consume more video and music traffic through their mobile devices (video will be $82\%$ of the total traffic by 2021~\cite{networking2016forecast}). 



Video (and music) traffic is highly skewed~\cite{youtube-traffic-from-edge,top-video-cellular,pptv-mobile-vod}, with a relatively small number of popular videos creating a large part of the demand. This skewness favors caching efficiency, and has made Content Distribution Networks (CDNs) a major component of today's wired Internet. Nevertheless, these ideas are not immediately applicable to future wireless networks. Even putting aside the interference-prone nature of the wireless channel, several factors limit the gains from mobile edge caching, in practice: (i) The size of content catalogues (even of only the most popular contents) is typically much larger than the storage capacity of a small base station. As a result, a large number of requests will always have to be redirected through the backhaul no matter what caching policy is used~\cite{Paschos-infocom2016, ElayoubiRoberts15}. (ii) The smaller coverage area of SCs implies frequent changes in the set of associated users in a SC, and this introduces larger variation in content request patterns. Thus, it becomes more difficult for a caching algorithm to accurately predict the ``local content popularity''~\cite{Paschos-infocom2016}. (iii) Viral videos routinely receive hundreds of thousands of views within a few hours after their upload~\cite{wikipedia-viral-videos} (and usually fade in equally fast). To accommodate for such trending videos, frequent updates of cached content is needed, which imposes a heavy load for the backhaul.

A large number of \textit{cooperative caching} policies have been recently proposed aiming to improve wireless caching efficiency, by exploiting coverage overlaps of small-cells~\cite{femto,poularakis2014toc}, coded transmissions~\cite{Ali2014}, caching on user equipment~\cite{Hui-offloading,sermpezis2014,Pavlos-Offload2016} and vehicles~\cite{Whitbeck-offloading,vigneri2016}, and others~\cite{collaborative-multi-bitrate-caching}. While these policies do offer considerable benefits \textit{in theory}, and under certain conditions some practical gains as well, they also face the above obstacles (as well as some additional ones, e.g., subpacketization complexity~\cite{shanmugam2016finite}, limited resources in the case of device-side caching, etc.) 

The above discussion suggests that caching policy alone is limited in the amount of performance gain it can bring at a small edge cache. Increasing cache capacity (to improve hit rates) or backhaul capacity (to allow for more frequent updates) seem like the only way to cope with this problem, but these are ``hardware'' solutions involving significant CAPEX/OPEX costs, when considering the very large number of small base stations envisioned in future heterogeneous and ultra-dense networks. The following question then arises: 
\textit{Are there any practical ``software-based'' solutions that can improve caching efficiency, at a low cost?}

Our answer to this question follows from two key observations: (i) the performance of a caching algorithm is dependent on user request patterns; (ii) user requests are increasingly driven by recommendation algorithms~\cite{cheng2009nettube,gupta2017modeling,RecImpact-IMC10}. For example, Netflix reports that around 80\% of its video views are through recommendations~\cite{gomez2016netflix}, while the corresponding percentage for YouTube's related video is 50\%~\cite{RecImpact-IMC10}. Our proposal is therefore to not try to further improve what is stored at each cache, but rather \textit{to better exploit the already cached content} by taking advantage of the recommendation algorithms, integrated in many of the content services that dominate the traffic consumption (e.g., YouTube, Netflix, Spotify). For instance, upon peak hours, when the backhaul is congested, a recommendation system could put higher preference on recommending content pre-cached in the vicinity of a user (e.g., at the base station with which she is associated, or to a nearby cache). 

However, recommendations cannot be arbitrarily manipulated in practice. The recommendation system still needs to accomplish its primary goal, i.e., recommend to users contents of their interest. Hence, the goal of this paper is to \textit{leverage the recommendation system to increase the caching efficiency, while at the same time ensuring high quality recommendations}. Note also that, recommending a content that is almost as interesting to the user but locally cached might not just be ``acceptable to the user, better for the network'', but even \textit{beneficial to both the user and the network}, if that content can be streamed for example at better quality (since congested links are avoided). 

The idea of improving edge caching by taking into account recommendations has been recently proposed~\cite{sermpezis-sch-globecom,sch-chants-2016,chatzieleftheriou2017caching,liu2018learning}, by optimizing caching policies based on~\cite{sermpezis-sch-globecom,sch-chants-2016} or jointly with recommendation systems~\cite{chatzieleftheriou2017caching,liu2018learning}. These solutions require control over the caching policy (e.g., if content and network provider are the same entities or collaborate), which might not be feasible for every scenario. What is more, these works assume simple models of content access, such as the Independent Reference Model (IRM), where consecutive requests by a user are independent.
In this paper, we go beyond the state of the art, by making the following contributions:
\begin{itemize}[leftmargin=*]
\item To our best knowledge, we propose the first model for recommendation-driven sequential user requests%
, that better fits real users behavior in a number of popular applications (e.g. YouTube, Vimeo, personalized radio). We then formulate the problem of maximizing the impact of recommendations on cache performance, while maintaining high quality for related contents (Sections~\ref{sec:problem-setup} and~\ref{sec:optimization-problem}). 
\item We show that the problem is non-convex, and propose an ADMM-like iterative algorithm that fully exploits the structure of the sequential content access statistics (\secref{sec:optim_algorithm}). 
\item Using simulations on a number of real datasets for different content types, we show that our algorithm significantly outperforms baseline approaches and previous work, improving caching performance in a large range of setups, while maintaining the desired recommendation quality (\secref{sec:sims}).
\end{itemize}

As a final note, while in this paper we consider the problem in the simple(r) context of a layer of caches with increasing cost, the proposed approach can be applied on top of a number of the cooperative wireless caching policies discussed earlier, e.g., ~\cite{femto,Ali2014,sermpezis-sch-globecom}. We plan to explore the interplay of such advanced caching policies and the proposed cache-friendly recommender as part of future work.

\section{Problem Setup}
\label{sec:problem-setup}

\myitem{Content Traffic.} We consider a content catalogue $\mathcal{K}$ of cardinality $K$, corresponding to a specific application (e.g. YouTube). A user can request a content from this catalogue either by asking \textit{directly} for the specific content (e.g., in a search bar) or by following a \textit{recommendation} of the provider. In practice, users spend on average a long time using such applications, e.g., viewing several related videos (e.g., 40~min. at YouTube~\cite{businessYoutubeSessions}), or listening to personalized radio while travelling.
%

\myitem{Recommendation System.} Recommendation systems have been a prolific area of research in the past years, and often combine content features, user preferences, and context with one or more sophisticated methods to predict user-item scores, such as collaborative filtering~\cite{sarwar2001item}, matrix factorization~\cite{koren2009matrix}, deep neural networks~\cite{covington2016deep}, etc. We will assume for simplicity that the \textit{baseline} recommender system (RS) for applications where the user consumes multiple contents works as follows: 

(i) The RS calculates a similarity score $u_{ij}$ between every content $i,j \in \mathcal{K}$, based on some state-of-the-art method; this defines a similarity matrix $U\in \mathbb{R}^{K \times K}$. Without loss of generality, let $u_{ij}\in[0,1]$, where we normalize values so that $u_{ij}=0$ denotes unrelated contents and $u_{ij}\rightarrow1$ ``very related contents''. W.l.o.g. we set $u_{ii}=0, \forall i \in \mathcal{K}$ for all contents. Note also that this $U$ might differ per user. 

(ii) After a user has just watched content $i$, the RS recommends the $N$ contents with the highest $u_{ij}$ value~\cite{RecImpact-IMC10,covington2016deep}. $N$ is usually a small number (e.g. values of $3-5$ are typical for the default YouTube mobile app) or sometimes $N=1$, as in the case of personalized radio (Spotify, last.fm) or ``AutoPlay'' feature in YouTube where the next content is simply sent to the user automatically by the recommender. 


\myitem{Caching Cost.} We assume that fetching content $i$ is associated with a cost $x_{i}\in\mathbb{R}$, which is known to the content provider. This cost might correspond to the delay experienced by the user, the added load in the backhaul network, or even monetary cost (e.g. for an Over-The-Top content provider leasing the infrastructure). It can also be used to capture different caching topologies. For example, to simply maximize the cache hit rate, we could set $x_i = 0$ for cached content, and $x_i =1$ for non-cached. For hierarchical caching~\cite{poularakis2014toc,borst2010}, the cost increases if the content is cached deeper inside the network.

Finally, as mentioned earlier, the specific wireless setup is relatively orthogonal to our approach and beyond the scope of this work. However, as a simple example, consider the well-known femto-caching setup~\cite{femto}. The proposed algorithm there would first decide what will be cached at each base station. Then, $x_{i}$ would have a low value for all content that the user in question can fetch from some BS in range (possibly dependent on the SINR of the BS, as well~\cite{femto}), and a high value otherwise.  

\myitem{User Request Model.} Based on the above setup, we assume the following content request model. 
\begin{definition}[User Request Model]\label{def:user-request}
After a user has consumed a content $i$, then
\begin{itemize}
\item (\textit{recommended request}) with probability $a$ the user picks one of the $N$ recommended contents with equal probability $\frac{1}{N}$.
\item (\textit{direct request}) with probability $1-a$ it ignores the recommender, and picks any content $j$ from the catalogue with probability $p_j$,  where $p_{j}\in [0,1]$ and $\sum_{j=1}^{K} p_{j}=1$. 
\end{itemize}
\end{definition}
$p_j$ above represents an underlying (long-term) popularity of content $j$, over the entire content catalogue. For short, we denote the vector $\mathbf{p_0}=[p_1,\dots,p_K]^{T}$. Note that the above model can easily generalized to consider different probabilities to follow different recommended contents (e.g. based on their ranking on the recommended list). Note also the assumption that $a$ is fixed: for instance,
for applications where the user cannot evaluate the content quality before she actually consumes the content, this assumption is realistic, at least in the ``short term''. In the remainder of the paper, we assume that if the recommendation quality is above a threshold, then the user's trust in the recommender (i.e. the value of $a$) remains fixed. We plan to explore scenarios where $a$ changes at every step, as a function of recommendation quality, in future work.

\myitem{Recommendation Control.} Our goal is to modify the user's choices through the ``recommended request'' part above, by appropriately selecting the $N$ recommended items. Specifically, let an indicator variable $z_{ij}$ denote whether content $j$ is in the list of $N$ recommended contents, after the user has watched content $i$. If $z_{ij} \in \{0,1\}$, the problem would be combinatorial and in most cases NP-hard. We can relax this assumption by letting $z_{ij} \in [0,1]$, and $\sum_{j} z_{ij} = N, \forall i$. $z_{ij}$ can be interpreted now as a probability. For example, if  $z_{13} =0.5$, then content 3 will be recommended half the times after the user consumes content 1. To facilitate our analysis we can further normalize this by defining variables $y_{ij} = \frac{z_{ij}}{N}$, $y_{ij}\in[0,\frac{1}{N}]$. It is easy to see that $y_{ij}$ define a stochastic matrix $Y$. Putting everything together, the above user request model can be defined by a Markov chain, whose transition matrix $P$ is given by
\begin{equation} \label{eq:Google-Markov-Matrix}
P = a\cdot Y+ (1-a)\cdot P_0,
\end{equation}
where $P_0 = \mathbf{1} \cdot \mathbf{p_{0}^T}$ is a rank-1 matrix ($P_0 \in \mathbb{R}^{K\times K}$), equal to the \textit{outer product} of a vector with $K$ unit values and the direct request vector $\mathbf{p_0}$. The above model of content requests, and the corresponding Markov Chain, is reminiscent of the well-known ``PageRank model''~\cite{page1999pagerank}, where a web surfer either visits an arbitrary webpage $i$ (with a probability $p_{i}$) or is directed to a webpage $j$ through a link from a webpage $i$ (with probability $p_{ij}$).

Table~\ref{table:notation} summarizes some important notation.

\begin{table}[h]
\centering
\caption{\textsc{Important Notation}}\label{table:notation}
\begin{small}
\begin{tabular}{|l|l|}
\hline
{$\mathcal{K}$}		&{Content catalogue (of cardinality $K$)}\\
\hline
{$u_{ij}$}			&{Similarity score for content pair $\{i,j\}$}\\
\hline
{$N$}			&{Number of recommended contents after a viewing}\\
\hline
{$x_{i}$}			&{Cost for fetching content $i$}\\
\hline
{$a$}			&{Prob. the user requests a recommended content}\\
\hline
{$p_{j}$}			&{Average a priori popularity of content $j$}\\
\hline
{$\boldsymbol{p_0}$}&{A priori popularity distribution of contents, $\in \mathbb{R}^K$}\\
\hline
{$z_{ij}$}			&{Prob. the RS recommends content $i$ after viewing $j$}\\
\hline
{$y_{ij}$}			&{Normalized prob. $y_{ij} = \frac{z_{ij}}{N}$}\\
\hline
{$\boldsymbol{\pi}$}	&{Stationary distribution of contents, $\in \mathbb{R}^K$}\\
\hline
{$\mathcal{C}$}		&{Set of cached content (of cardinality $C$)}\\
\hline
\end{tabular}
\end{small}
\end{table}




\section{Problem Formulation}\label{sec:optimization-problem}
Given the above setup, our general goal in this paper is \textit{to reduce the total cost of serving user requests by choosing matrix $Y$, while maintaining a required recommendation quality}.

Consider a user that starts a session by requesting a content $i\in\mathcal{K}$ with probability $p_{i}$ (i.e., we assume her initial choice is not affected by the recommender), and then proceeds to request a sequence of contents according to the Markov chain $P$ of Eq.(\ref{eq:Google-Markov-Matrix}). Assume that the user requests $M$ contents in sequence. Then the associated access cost would be given by
\begin{equation}\label{cost-finite-steps}
\sum_{m=0}^{M} \mathbf{p_0}^{T}\cdot P^{m} \cdot \mathbf{x},
\end{equation}
where $\mathbf{x} = [x_{1},...,x_{K}]^{T}$ is the vector of the costs per content (see \secref{sec:problem-setup}).

$M$ is a random variable though, and the various powers of transition matrix $P$, which contains the control variable $Y$, would greatly complicate the problem. However, the above Markov chain is strongly connected and ergodic under very mild assumptions for $\mathbf{p_0}$. It thus has a stationary distribution $\boldsymbol{\pi} = [\pi_{1},...,\pi_{K}]^{T}$, which is also equal to the long-term percentage of total requests for content $i$. Consequently, for $M$ large enough we can approximate the average cost \textit{per request} with
\begin{equation}\label{eq:definition-average-cost}
\boldsymbol{\pi}^{T}\cdot \mathbf{x}
\end{equation}
where $\boldsymbol{\pi}$ can be calculated from the following lemma.

\begin{lemma}\label{lemma:stationary-analytical-expression}
The stationary distribution $\boldsymbol{\pi}$ is given by
\begin{equation}\label{eq:stationary-analytical-expression}
\boldsymbol{\pi}^{T} = (1-a)\cdot  \mathbf{p_0}^{T} \cdot (I - a\cdot Y)^{-1} 
\end{equation}
where $I$ the $K\times K$ identity matrix.
\end{lemma}
\begin{proof}
The stationary distribution above can be derived through the standard stationary equality~\cite{mor2013}
\begin{equation}
\boldsymbol{\pi}^{T} =a \cdot \boldsymbol{\pi}^{T} \cdot Y + (1-a) \cdot  \mathbf{p_0}^{T}, 
\end{equation}
by observing that matrix $(I - a\cdot Y)$ has strictly positive eigenvalues (in measure). See also~\cite{avrachenkov2006pagerank}, for more details.
\end{proof}

We are therefore ready to formulate cache-friendly recommendations as an optimization problem.

\begin{problem}[Cache-Friendly Recommendations]\label{problem:BASIS-optimization-problem}
\begin{small}
\begin{subequations}\label{eq:objective-infinite-step}
\begin{align}
\underset{Y}{\mbox{minimize}} \; \; \; \; & 
  \mathbf{p_{0}}^{T} \cdot (I-aY)^{-1} \cdot \mathbf{x}, \tag{\ref{eq:objective-infinite-step}}\\
& 0 \leq y_{ij} \leq \frac{1}{N}, ~~\forall~i~and~j~\in~\mathcal{K}. \label{eq:y-box-constraint} \\
& \sum_{j =1}^{K}   y_{ij} = 1, ~~\forall i~\in \mathcal{K} \label{eq:sum-y-equals-1-constraint}\\
& y_{ii} = 0,~~\forall~i~\in \mathcal{K} \label{eq:no-self-recommendations-constraint}\\
& \sum_{j =1}^{K}   y_{ij} u_{ij} \geq q_i, ~~\forall i~\in \mathcal{K} \label{eq:quality-constraint}
\end{align}
\end{subequations}
\end{small}
\end{problem}

\myitemit{Objective.} The objective is to minimize the expected cost to access any content, and follows directly from \eq{eq:definition-average-cost} and Lemma~\ref{lemma:stationary-analytical-expression}. Note that we have dropped the constant $(1-a)$ from \eq{eq:stationary-analytical-expression}, as it does not affect the optimal solution. 

\myitemit{Control Variables.} The variables $y_{ij}$ ($K^{2}$ in total), deciding what is recommended after each content $i$, constitute the control variables. 

\myitemit{Constraints.} The first three constraints make sure that $y_{ij}$ forms a stochastic transition matrix that can be translated to $N$ recommendations per item $i$. Specifically, the ``box'' constraints of \eq{eq:y-box-constraint} ensures that all entries are positive, and smaller or equal to $1/N$. The latter is necessary as matrix $Y$ is a \textit{normalized} version of the recommendations. Together with \eq{eq:sum-y-equals-1-constraint} these ensure that exactly $N$ contents are recommended for every $i$ (see also Section~\ref{sec:problem-setup}, ``Recommendation Control''). \eq{eq:no-self-recommendations-constraint} simply ensures that the same content cannot be recommended when it was just consumed.

\myitemit{Quality Constraint.} \eq{eq:quality-constraint} ensures that that the ``quality'' of recommended contents for each $i$ is above a desired threshold. Observe that, without this constraint, the optimal solution to the above problem is trivial, namely to always recommend the same $N$ contents $j$ with the minimum cost $x_{j}$. However, these contents will probably be unrelated (i.e. $u_{ij} \rightarrow 0$) essentially ``breaking'' the recommender. Hence, this constraint forces variables $y_{ij}$ to select high $u_{ij}$ values to ensure the recommender keeps doing its primary job, namely finding related contents. Note that, if there are at least $N$ strongly related contents for each $i$ (i.e., $u_{ij} = 1$), then the maximum value for $q_i$ is $1$. W.l.o.g., in the remainder we will assume the same quality constraint for all $i$ ($q_i = q$). 

The remaining quantities are constants, and inputs to the problem. While some of them might still vary over time (e.g., $u_{ij}$), we assume this occurs at a larger time scale, compared to our problem.

Unfortunately, the objective function~(\eq{eq:objective-infinite-step}) is non-convex, unless $Y$ is positive semidefinite and symmetric. In that case, the problem could be cast into an SDP (Semi-Definite Program) using Schur's complement~\cite{boyd2004convex}. However, forcing $Y$ to be symmetric in our problem leads to trivial solutions, as every symmetric Markov chain has a uniform invariant measure. What is more, the inverse in the objective further complicates solving this problem, as the gradient of this expression is rather complex. 

\section{Optimization Algorithm}
\label{sec:optim_algorithm}
Given that \probref{problem:BASIS-optimization-problem} is non-convex, there are no polynomial-time algorithms that can guarantee to converge to the optimal solution. This leaves us with two options for solving the problem: to apply (i) an exponential-time ``global'' optimization algorithm (e.g., Branch-and-Bound), or (ii) a heuristic algorithm for an approximate solution. The former is infeasible for all practical scenarios, due to the large problem size ($K^{2}$ control variables). Therefore, we will consider two heuristic approaches: in \secref{sec:greedy-algorithm} we consider a ``myopic'' algorithm, essentially a greedy approach that solves a simpler objective than \probref{problem:BASIS-optimization-problem}; this algorithm will be our baseline, as it resembles some recent state-of-the-art~\cite{chatzieleftheriou2017caching}); in \secref{sec:admm-algorithm}, we propose a more sophisticated algorithm, inspired from ADMM type of schemes~\cite{boyd2011distributed}.

\subsection{Myopic Algorithm}\label{sec:greedy-algorithm}

The non-convexity of \probref{problem:BASIS-optimization-problem} is due to the expression of the stationary distribution $\boldsymbol{\pi}$ that appears in the objective function. As mentioned, the stationary distribution captures the long-term behavior of a system where users sequentially consume many contents. To simplify the objective, one can could consider a coarse approximation where the recommendation impact is there, but the algorithm ``greedily'' optimizes the access cost \textit{only for the next content access}. In other words, it is as if  a user initially requests a content $i$, then requests another content $j$ (recommended or not), and then leaves the system. In this case, the objective becomes
\begin{equation}
(\mathbf{p_{0}}^{T}\cdot P) \cdot \mathbf{x},
\end{equation}
where the first term of \eq{cost-finite-steps} is dropped (because it is independent of the control variables), and we keep only the second term. This gives rise to the following optimization problem.

\begin{problem}[Myopic Cache-Friendly Recommendations]\label{problem:single-step-optimization}
\begin{small}
\begin{align}
\underset{Y}{\mbox{minimize}}~~ \; & 
\mathbf{p_{0}}^{T} \cdot (a\cdot Y + (1-a)\cdot P_0) \cdot x, \label{eq:objective-single-step}\\
s.t.~~~ &\text{Eqs. (\ref{eq:y-box-constraint})--(\ref{eq:quality-constraint})} \nonumber 
\label{eq:constraint-set-single-step}
\end{align}
\end{small}
\end{problem}

In the above problem, the constraints remain intact as the Problem~\ref{problem:BASIS-optimization-problem}. However, now the objective is linear in $Y$. This is an Linear Problem (LP) with affine and box constraints, which can be solved efficiently in polynomial time, using e.g. interior-point methods~\cite{boyd2004convex}. 

\myitemit{Remark.} The single-step approach can be interpreted as a projection of the recent work of~\cite{chatzieleftheriou2017caching} to our framework. Specifically, the authors solve a similar ``single-step'' problem, jointly optimizing the caching and recommendation policy (which is formulated as a Knapsack problem). Omitting the caching decisions of~\cite{chatzieleftheriou2017caching}, for the recommendations the authors solve a similar problem to \probref{problem:single-step-optimization}.


\oldtext{
The problem constraints remain intact as the Problem~\ref{problem:BASIS-optimization-problem}, however now the objective is linear in $Y$, therefore, we have in hand an LP with affine and box constraints, thus trivially solved. 
Essentially, the work described in~\cite{chatzieleftheriou2017caching}, and more specifically the part of recommendations, if projected onto our framework, can be parallelized using the above problem.
}

\subsection{Cache-Aware Recommendations for Sequential content access (CARS)}\label{sec:admm-algorithm}

The above ``myopic'' approach does not exploit the full structure of the Markov chain $P$. For example, assume there are two contents $A$ and $B$ that are both cached and both have high similarity with a content currently consumed, but $B$ has slightly higher similarity. The Myopic scheme will choose to recommend $B$. However, assume that $A$ is similar to many contents that happen to be cached, while $B$ does not. This suggests that, if $B$ is recommended, then in the next step there will be very few good options (hence the algorithm's name): the myopic algorithm will either have to recommend cached contents with low quality or high quality contents which lead to cache misses. To be able to foresee such situations and take the right decisions, we need to go back to the objective of \probref{problem:BASIS-optimization-problem}. 

To circumvent the problem of having the inverse of the control matrix in the objective, we formulate an \textit{equivalent} optimization problem by introducing the stationary vector $\boldsymbol{\pi}$ as an explicit (``auxiliary'') control variable.

\begin{problem}[Cache-Friendly Recommendations: Equivalent Problem]\label{problem:infinite-step-admm}
\begin{small}
\begin{subequations}\label{eq:equivalent-objective-infinite-step}
\begin{align}
\underset{\boldsymbol{\pi},Y}{\mbox{minimize}}~~ \; & 
 \boldsymbol{\pi}^{T}\cdot \mathbf{x}, \tag{\ref{eq:equivalent-objective-infinite-step}}\\
s.t.~~~ &\text{Eqs. (\ref{eq:y-box-constraint})--(\ref{eq:quality-constraint})} \nonumber \\
& \boldsymbol{\pi}^{T} = \boldsymbol{\pi}^{T} \cdot (a\cdot Y+(1-a)\cdot \mathbf{p_{0}}^{T})\label{eq:stationarity-hard-con}\\
& \sum_{j =1}^{K} \pi_{j} = 1 \label{eq:stationarity-sum-pi}\\
& \pi_{j} \geq 0,~~\forall~j~\in \mathcal{K}. \label{eq:stationarity-positive-pi}
\end{align}
\end{subequations}
\end{small}
\end{problem}


\probref{problem:infinite-step-admm} constraints three new (sets of) constrains. \eq{eq:stationarity-sum-pi} and \eq{eq:stationarity-positive-pi} simply ensure that $\boldsymbol{\pi}$ is a probability distribution. \eq{eq:stationarity-hard-con} is an important constraint that ensures that the two problems are equivalent, by forcing  $\boldsymbol{\pi}$ to be a stationary distribution related to the transition matrix $P = a\cdot Y+(1-a)\cdot P_{0}$. It is easy to see that the two problems have the same set of optimal solutions.

The objective function is now linear in the control variables $\boldsymbol{\pi}$. However, constraint \eq{eq:stationarity-hard-con} is a quadratic equality constraint, and thus the problem remains non-convex. Nevertheless, observe that the problem is now \textit{bi-convex} in the variables $Y$ and $\boldsymbol{\pi}$. Bi-convex problems can often be efficiently tackled with Alternating Convex Search (ACS) methods, that iteratively solve the convex sub-problems for each set of control variables. Unfortunately, such approaches fail here, as the $Y$ subproblem is simply a feasibility problem ($Y$ does not appear in the objective), and ACS would not converge (our implementation confirms this observation). What is more, having the quadratic equality constraint as a hard constraint does not facilitate such an iterative solution.

Instead, we propose to use a Lagrangian relaxation for that constraint, moving it to the objective. To ensure the strong convexity of the new objective, we form the \textit{Augmented Lagrangian}~\cite{boyd2011distributed}. Let us first define the function $c(\boldsymbol{\pi},Y)$ as
\begin{equation}
c(\boldsymbol{\pi},Y) = \boldsymbol{\pi}^T - \boldsymbol{\pi}^T \cdot (a\cdot Y - (1-a)\cdot P_{0})
\end{equation}
so that the constraints of \eq{eq:stationarity-hard-con} can be written as
\begin{equation}
c(\boldsymbol{\pi},Y) = 0
\end{equation}

The augmented Lagrangian is then given by:
\begin{equation}\label{eq:equivalent-objective-infinite-step-AL}
f_{\rho}(\boldsymbol{\pi},Y) = \boldsymbol{\pi}^{T}\cdot x  + c(\boldsymbol{\pi},Y) \cdot \boldsymbol{\lambda}+ \frac{\rho}{2}\cdot (||c(\boldsymbol{\pi},Y)||_{2})^{2}
\end{equation}
where $\boldsymbol{\lambda}$ is the column vector of length $K$ of the Lagrangian multipliers (one multiplier per quadratic equality), $\rho$ a positive constant scalar, and $||\cdot ||_{2}$ the euclidean norm. This objective is still subject to the remaining constraints of \probref{problem:infinite-step-admm}, all of which are now affine. What is more, the problem remains bi-convex in the control variables $Y$ and $\boldsymbol{\pi}$. We can thus apply an ADMM-like method, where we iteratively solve the convex subproblems with respect to $Y$ and $\boldsymbol{\pi}$, but now with the above augmented objective, so that when $c(\boldsymbol{\pi},Y)$ diverges a lot from $0$, the subproblem solutions in the inner loop are penalized. We also update the Lagrangian multipliers $\lambda_i$ at each iteration. Our detailed algorithm is described in Algorithm~\ref{alg:ADMM}.


\begin{algorithm}
\begin{algorithmic} [1]
\caption{CARS (Cache-Aware Recommendations for Sequential content access)}\label{alg:ADMM}
\Statex {$Input: N, U, q, \mathbf{x}, a, \mathbf{p_{0}}$} \Comment (system parameters)
\Statex {$Input: Acc_1, Acc_2, maxIter, \rho, \lambda_0, Y_0$} \Statex\Comment (algorithm tuning parameters)
\State {$i \gets 1$}
\State $COST_0 \gets \infty$
\State $V \gets True$
\While{$V$} 
\State $\boldsymbol{\pi}_i = \underset{\boldsymbol{\pi}\in C_{\pi}}{\operatorname{argmin}} \{f_{\rho}(\boldsymbol{\pi},Y_{i-1}) \}$
\State $Y_i = \underset{Y\in C_{Y}}{\operatorname{argmin}} \{f_{\rho}(\boldsymbol{\pi}_{i},Y) \}$
\State $\lambda \gets \lambda + (\frac{\rho}{2})\cdot c(\boldsymbol{\pi}_{i},Y_{i})$
\State $COST_i \gets (1-a)\cdot \mathbf{p_{0}^{T}} \cdot (I-a\cdot Y_{i})^{-1} \cdot \mathbf{x}$
\State $\epsilon_1 \gets (|c(\pi_i,Y_i)|_{2})^{2}$
\State $\epsilon_2 \gets |COST_i - COST_{i-1}|$
\State $V = ((\epsilon_1>Acc_1)\land(\epsilon_2>Acc_2))\lor(i \leq maxIter)$
\State $i \gets i+1$
\EndWhile
\State $j\gets \underset{\ell=1,...,i-1}{\operatorname{argmax}} \{COST_{\ell}\}$
\State $return~~Y_{j}$

\end{algorithmic}
\end{algorithm}

Algorithm~\ref{alg:ADMM} receives as input the system parameters $N, U, q, \mathbf{x}, a, \mathbf{p_{0}}$, and the desired accuracy levels and initialization parameters $Acc_1, Acc_2, maxIter, \rho, \lambda_0, Y_0$. It initializes the objective ($COST_{0}$) to infinity and starts an iteration for solving the convex subproblems (lines 4--13). In the first leg of the loop (line 5), the augmented Lagrangian $f_{\rho}(\boldsymbol{\pi},Y)$ is minimized over $\pi$, considering as constant the variables $Y$ (equal to their prior value). Then, considers the returned value of $\boldsymbol{\pi}$ from line 5 as constant and minimizes the Lagrangian over the variables $Y$. Both minimization sub-problems are convex and can be efficiently solved. The solution space of the sub-problems $C_{Y}$ and $C_{\pi}$ is given by Eqs.~(\ref{eq:y-box-constraint})--(\ref{eq:quality-constraint}) and Eqs.(\ref{eq:stationarity-sum-pi})--(\ref{eq:stationarity-positive-pi}), respectively. After calculating in line 8 the long term $COST$ we get from $Y_i$, the status of the current iteration is computed in the (a) primal residual of the problem (line 9) and (b) the difference of returned $COST$ compared to the previous step (line 10). The algorithm exits the while loop, when the value of the primal residual and improvement in the $COST$ are smaller than the required accuracy, or when the maximum allowable iterations are reached (as described in line 11). 

As a final note, the above problem can also be cast into a non-convex QCQP (quadratically constrained quadratic program). State-of-the-art heuristic methods for approximate solving generic QCQP problems~\cite{park2017general} are unfortunately of too high computational complexity for problems of this size. It is worth mentioning that we transformed the problem to a standard QCQP formulation and we applied methods based on~\cite{park2017general} but the algorithms were only capable of solving small instances of the problem (a few 10s of contents).

\myitem{Convergence of CARS.}
Finally, we investigate the performance of CARS (Algorithm~\ref{alg:ADMM}) as a function of its computational cost, i.e., the maximum number of iterations needed. Fig.~\ref{fig:admm-convergence} shows the achieved \textit{actual cost} (red line, circle markers) at each iteration, and the \textit{virtual cost} (gray line, triangle markers) calculated from the current value of the auxiliary variable $\boldsymbol{\pi}$, in a simulation scenario (see details in Section~\ref{sec:sims}). It can be seen that within $5$ iterations, CARS converges to its maximum achieved cache hit ratio. This is particularly important for cases with large content catalogue sizes that require an online implementation of CARS.

\begin{figure}
    \centering
        \includegraphics[width=0.6\columnwidth]{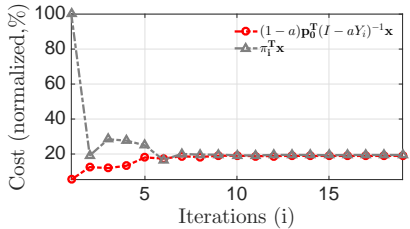}
        \caption{Convergence of CARS.}
        \label{fig:admm-convergence} 
\end{figure}

\section{Performance Evaluation}
\label{sec:sims}
In this section, we investigate the improvements in caching performance by the proposed cache-aware recommendation algorithm on top of a preselected caching allocation. We perform simulations using real datasets of related contents (movies and songs), collected from online databases. We first briefly present the datasets (\secref{sec:datasets}) and the simulation setup (\secref{sec:simulation-setup}), and then present simulation results in a wide range of scenarios and parameters and discuss the main findings and implications (\secref{sec:simulation-results})

\subsection{Datasets}\label{sec:datasets}
We collect two datasets that contain ratings about multimedia content. We use this information to build similarity matrices $U$, which are later used in the selection of recommendations, e.g., to satisfy a minimum recommendation quality $q$ (as defined in \secref{sec:optimization-problem}).


\myitem{MovieLens.} We use the $100k$ subset from the \textit{latest Movielens} movies-rating dataset from the MovieLens website~\cite{movielens-related-dataset}, containing $69162$ ratings (from 0.5 to 5 stars) of $671$ users for $9066$ movies. To generate the matrix $U$ of movie similarities from the raw information of user ratings, we apply a standard collaborative filtering method~\cite{survey-collaborative-filtering}. Specifically, we first apply an item-to-item collaborative filtering (using 10 most similar items) to predict the missing user ratings, and then use the cosine-distance ($\in[-1,1]$) of each pair of contents based on their common ratings
\begin{equation*}
sim(i,j) = \frac{\sum_{n=1}^{\#users} r_{n}(i)\cdot r_{n}(j)}{\sqrt{\sum_{n=1}^{\#users} r_{n}^{2}(i)} \cdot \sqrt{\sum_{n=1}^{\#users} r_{n}^{2}(j)}}
\end{equation*}
where we normalized the ratings $r_{i}$, by subtracting from each rating the average rating of that item. We build the matrix $U$ by saturating to values above 0.6 to 1, and zero otherwise, so that $u_{nk}\in \{0,1\}$.

\myitem{Last.fm.} We use the subset of \textit{The Million Song Dataset} from the Last.fm database~\cite{Bertin-Mahieux2011}, containing $10k$ song IDs. The dataset was built based on the method ``getSimilar'', and thus it contains a $K\times K$ matrix with the similarity scores (in [0,1]) between each pair of songs in the dataset, which we use as the matrix $U$. As the Last.fm dataset is quite sparse and we set the non zero values $u_{ij}$ to one to make a binary $U$ in that dataset as well. 

To facilitate simulations, we process both datasets, by removing rows and columns of the respective $U$ matrices with $\sum_{j \in \mathcal{K}}u_{ij} \leq N$ (where number $N=4$ is the number of total recommendations). After the preprocessing, we ended up with a content catalogue of size $K=1060$ and $K=757$ for MovieLens and Last.fm traces respectively.



\subsection{Simulation Setup}\label{sec:simulation-setup}
\myitem{Content Demand.} The users generate $40000$ requests for contents in a catalogue $\mathcal{K}$; requests are either \textit{direct} with probability $\mathbf{p_{0}} \sim Zipf(s)$ ($s$ the exponent of the Zipf law) for any content, or \textit{recommended} with probability $\frac{1}{N}$ for each of the recommended contents. We consider scenarios with exponent $s\in [0.4, 0.8]$ and $N=4$. Unless otherwise stated, we set the default value $\alpha=0.8$, similarly to the statistics in~\cite{gomez2016netflix}.

%
%

\myitem{Caching Policy.} We consider a popularity based caching policy, where the $C$ most popular (w.r.t. $\mathbf{p_{0}}$) contents are locally cached in the base station. This policy is optimal in a single cache network, when no recommendation system is employed.

\myitem{Recommendation policy.} We simulate scenarios under the following three recommendation policies:
\begin{itemize}[leftmargin=*]
\item \emph{No Recommendation}: This is is also a baseline scenario, where users request contents only based on $\mathbf{p_{0}}$ (or, equivalently $a=0$). \item \emph{Myopic policy}: Cache-aware recommendations using the algorithm of \secref{sec:greedy-algorithm}, which optimizes recommendations assuming single-step content requests. This policy relates to the previous works of~\cite{sermpezis-sch-globecom,chatzieleftheriou2017caching}.
\item \emph{Proposed Policy/CARS}: Cache-aware recommendations using \textit{CARS}, which optimizes recommendations for sequential content consumption.
\end{itemize}

\begin{figure}
\centering
\subfigure[MovieLens,~$s = 0.7$]{\includegraphics[width=0.6\columnwidth]{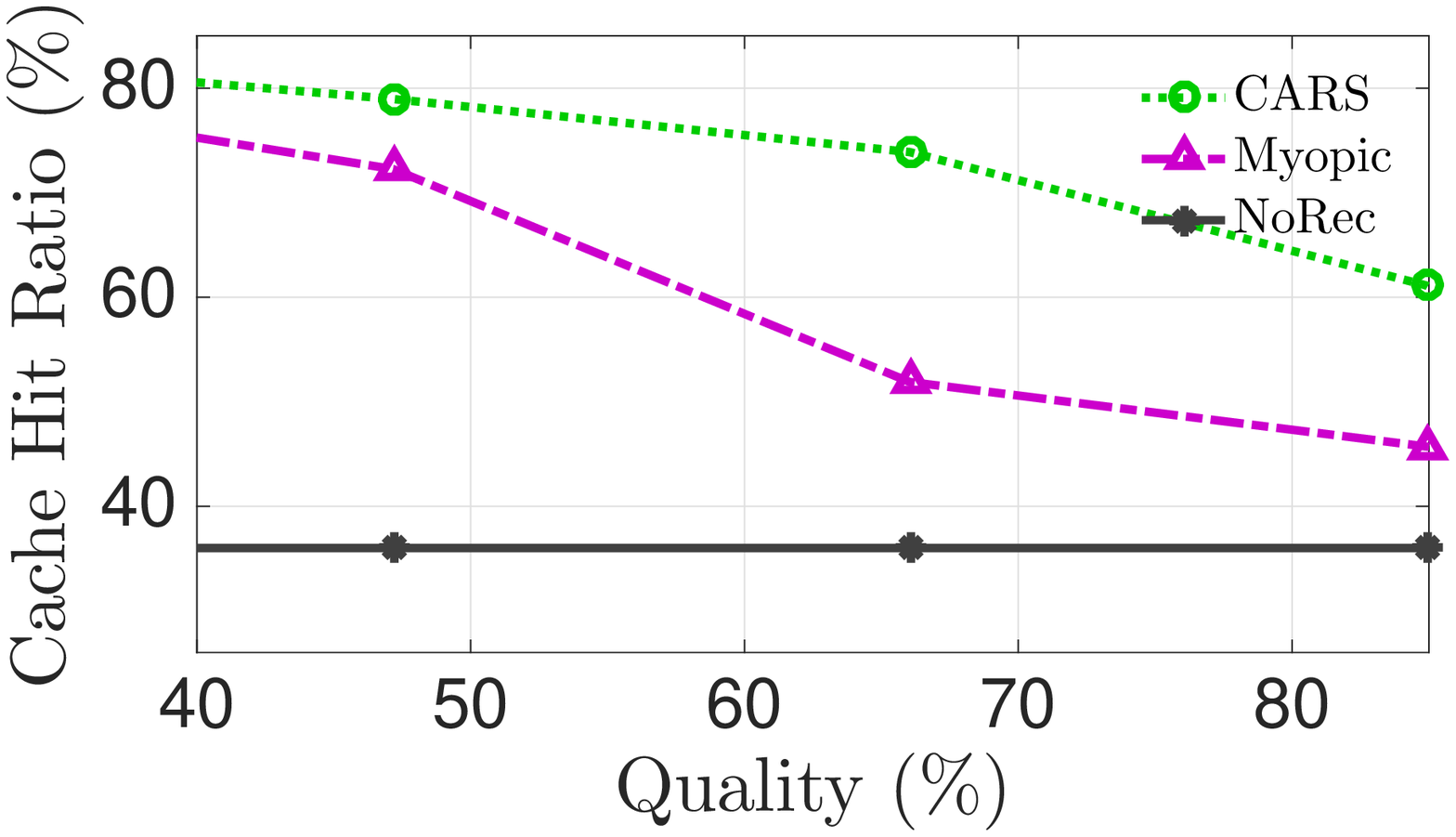}\label{fig:movieles-quality}}
\subfigure[Last.fm,~$s=0.4$]{\includegraphics[width=0.6\columnwidth]{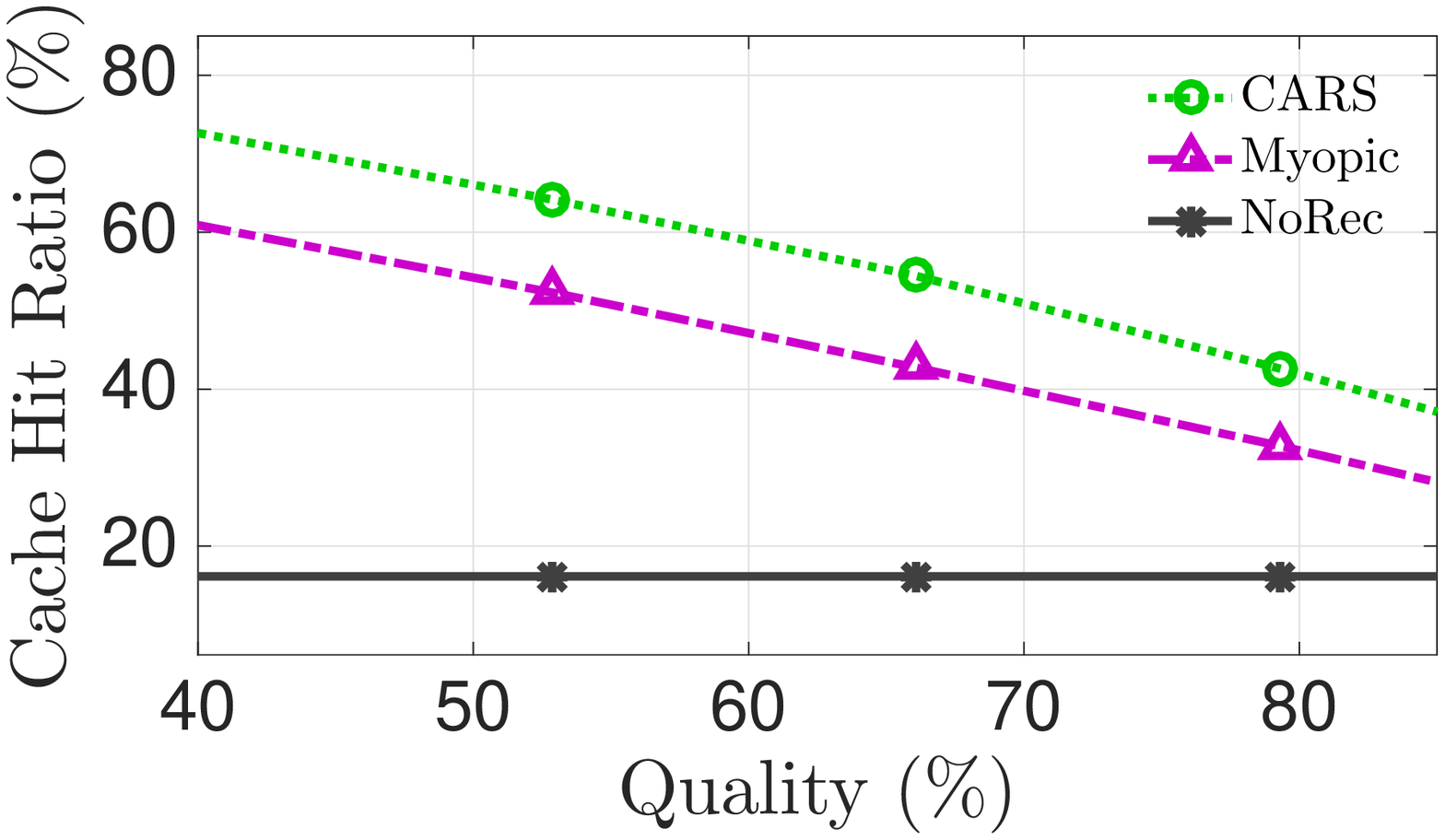}\label{fig:lastfm-quality}}
\caption{Cache Hit Ratio vs Quality~$N=4,~C/K=5\%$.}
\end{figure}

\subsection{Results}\label{sec:simulation-results}
We compare the three recommendation policies in scenarios with varying $q$ (minimum quality of recommendations - see \secref{sec:optimization-problem}), cache size $C$, probability to request a recommended content $a$, and $N$ recommended contents. For simplicity, we assume costs $x=0$ for cached contents and $x=1$ for non-cached. Hence, the cost becomes equivalent to the cache hit ratio ($CHR = (1-a)\cdot \mathbf{p_{0}}^{T} \cdot (I-aY)^{-1} \cdot \mathbf{x}$), which we use as metric to measure the achieved performance in our simulations.

\myitem{Impact of Quality of Recommendations.}
Recommending cached contents becomes trivial, if no quality in recommendations is required. However, the primary goal of a content provider is to satisfy its users, which translates to high quality recommendations. In the following, we present results that show that the proposed \textit{CARS can always achieve a good trade-off between cache hit ratio and quality of recommendation, significantly outperforming baseline approaches}.

%

In Figures~\ref{fig:movieles-quality} and~\ref{fig:lastfm-quality} we present the achieved cache hit ratio (y-axis) of the four recommendation policies for the MovieLens and Last.fm, datasets, respectively, in scenarios where the \text{recommender quality} is imposed to be above a predefined threshold $q$ (x-axis). The first observation is that \textit{Myopic} and \textit{CARS} achieve their goal to increase the CHR compared to the baseline case of \textit{NoRec}. The absolute gains for both policies increases for lower values of $q$, because for lower $q$ there is more flexibility in recommendations. For high values of $q$, close to $100\%$, less recommendations that ``show the cache'' are allowed, and this leads to lower gains. However, even when the quality reaches almost $100\%$, the gains of \textit{CARS} remain significant. In fact, the relative performance of \textit{CARS} over the \textit{Myopic} increases with $q$, which indicates that non-\textit{Myopic} policies are more efficient when high recommendation quality is required. 

Moreover, comparing Figures~\ref{fig:movieles-quality} and~\ref{fig:lastfm-quality} reveals that the achievable gains depend also on the similarity matrix $U$. While in Fig.~\ref{fig:lastfm-quality} both cache-aware recommendation policies follow a similar trend (for varying $q$), in Fig.~\ref{fig:movieles-quality} for the larger dataset of MovieLens, the performance of \textit{CARS} decreases much less compared to \textit{Myopic} with $q$. 




\begin{figure}
\centering
\subfigure[MovieLens,~$s=0.5$]{\includegraphics[width=0.6\columnwidth]{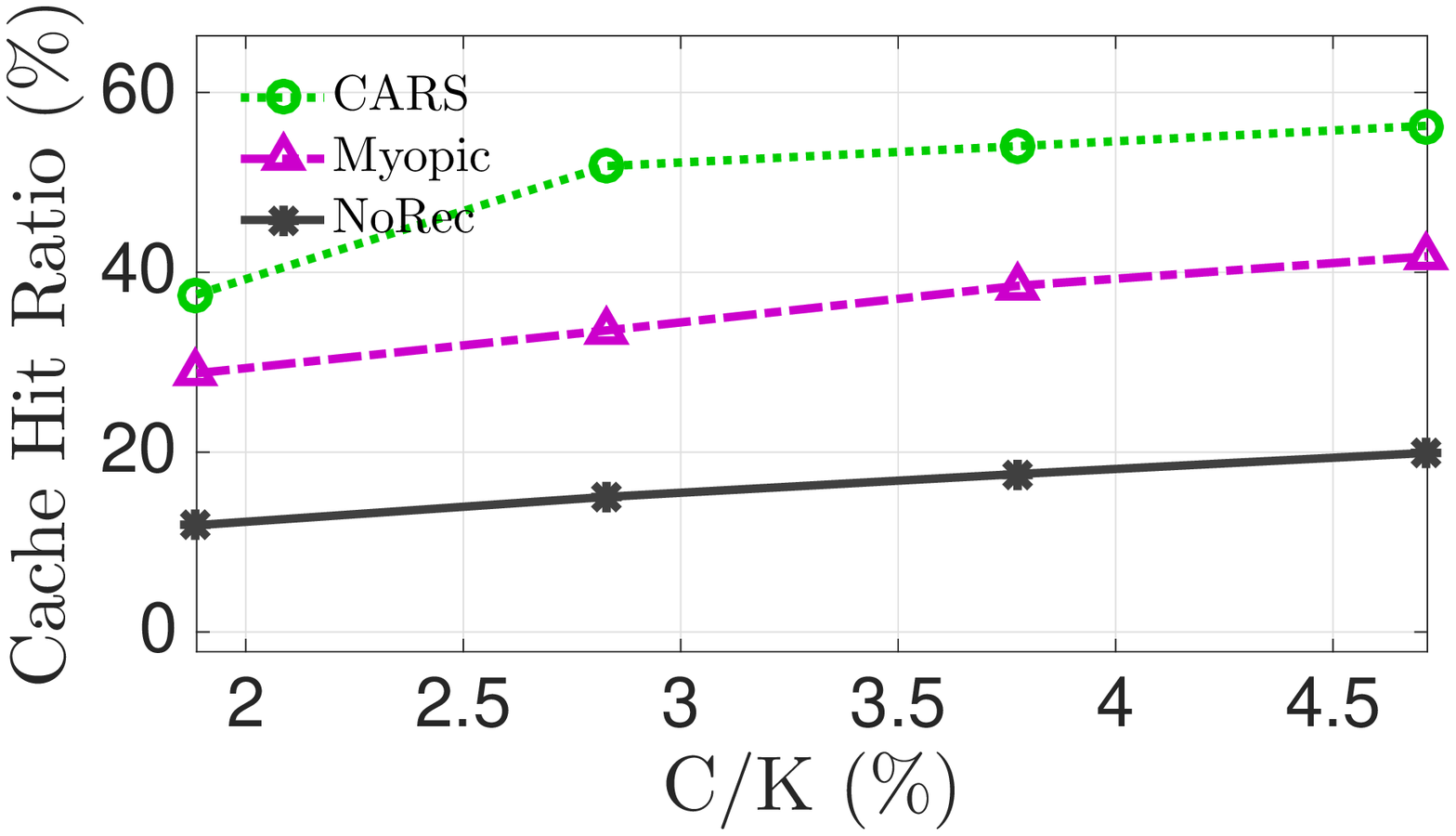}\label{fig:movieles-cache-size}}
\subfigure[Last.fm,~$s=0.4$]{\includegraphics[width=0.6\columnwidth]{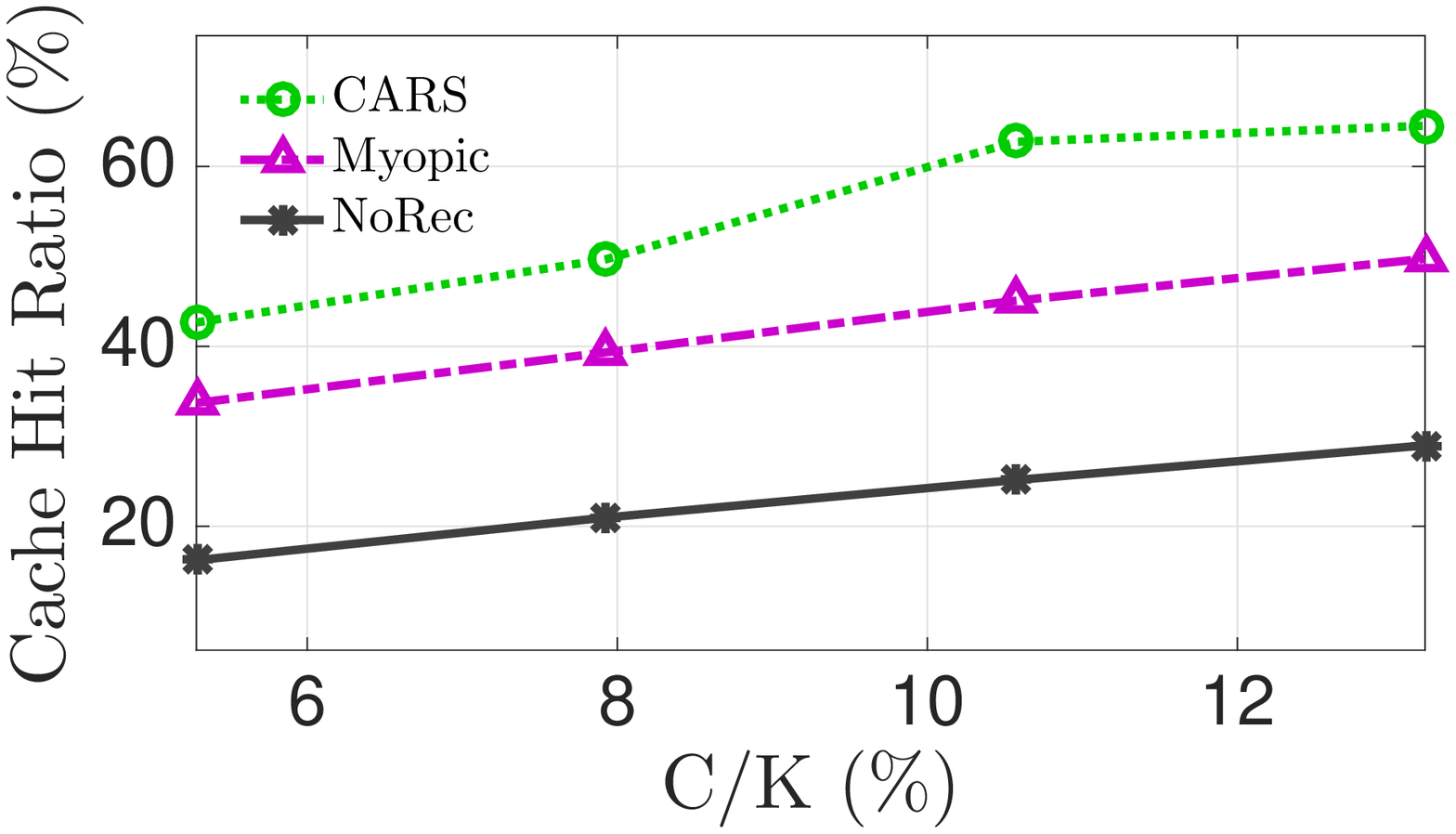}\label{fig:lastfm-cache-size}}
\caption{Cache Hit Ratio vs Relative Cache size,~$Q=80\%,~N=4$.}
\end{figure}

\myitem{Impact of Caching Capacity.}
In Figures~\ref{fig:movieles-cache-size} and~\ref{fig:lastfm-cache-size} we investigate the performance of the recommendation policies with respect to the cache size, for a fixed value of the recommender quality $q$. The proposed algorithm, outperforms significantly the other two policies. For example, in Fig.~\ref{fig:lastfm-cache-size}, for $C/K=8\%$ it achieves a $25\%$ improvement over the \textit{Myopic} algorithm. Even in the case of the MovieLens dataset (Fig.~\ref{fig:movieles-cache-size}), where the \textit{Myopic} algorithm can only marginally improve the cache hit ratio, \textit{CARS} still achieves significant gains. In total, in all scenarios we considered, \textit{the relative caching gains from the proposed cache-aware recommendation policy (over the no-recommendation case) are consistent and even increase with the caching size.}

\myitem{Impact of Sequential Content Consumption.}
\textit{CARS} takes into account the fact that users consume more than one content sequentially, and optimizes recommendations based on this. On the contrary the \textit{Myopic} algorithm (similarly to previous works~\cite{sermpezis-sch-globecom},~\cite{chatzieleftheriou2017caching}) considers single content requests. Therefore, Algorithm~\ref{alg:ADMM} is expected to perform better as the average number of consecutive requests by a user increases. The simulation results in Fig.~\ref{fig:num-access} validate this argument. We simulate a scenario of a small catalogue $K=100, C=4, N=3, s=0.6, q=90\%$ and a $U$ matrix with an $\overline{R}=4$ related contents per content on average, where we vary the number of consecutive requests by each user. It can be seen that the \textit{Myopic} algorithm increases the cache hit ratio when the users do a few consecutive requests (e.g., 3 or 4); after this point the cache hit ratio remains constant. However, under \textit{CARS}, not only the increase in the cache hit ratio is higher, but it increases as the number of consecutive requests increase. This is a promising message for real content services (such as YouTube, Netflix, Spotify, etc.) where users tend to consume sequentially many contents.

\myitem{Impact of Probability $a$.}
The probability $a$ represents the \textit{frequency} that a user follows a recommendation rather than requesting for an arbitrary content (restart probability, e.g., through the search bar in YouTube). The value of $a$ indicates the influence of the recommendation system to users; in the cases of YouTube and Netflix it is approximately 0.5 and 0.8 respectively~\cite{RecImpact-IMC10},~\cite{gomez2016netflix}. In Fig.~\ref{fig:restart} we present the performance of the two cache-aware recommendation policies for varying values of $a$. The higher the value of $a$, the more frequently a user follows a recommendation, and thus the higher the gains  from the cache-aware recommendation policies. However, while the gain from the \textit{Myopic} algorithm increases linearly with $a$, the gains from the proposed \textit{CARS} increase superlinearly. This is due to the fact that Algorithm~\ref{alg:ADMM} takes into account the effect of probability $a$ when selecting the recommendations (e.g., see the objective function of \probref{problem:infinite-step-admm}).

\begin{figure}
\centering
\subfigure[vs \# of Accesses (synthetic)]{\includegraphics[width=0.6\columnwidth]{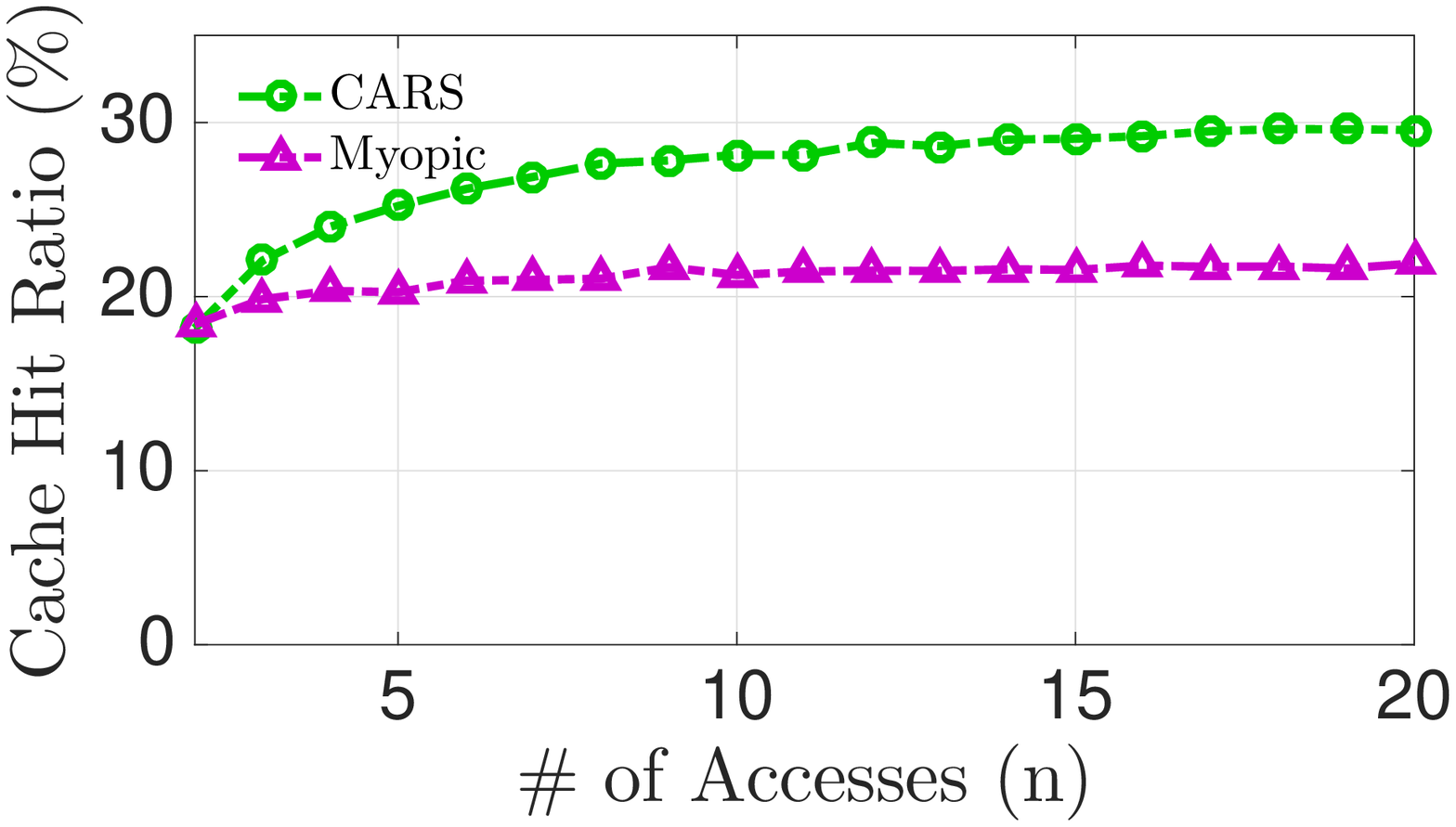}\label{fig:num-access}}
\subfigure[vs Probability (a) (synthetic)]{\includegraphics[width=0.6\columnwidth]{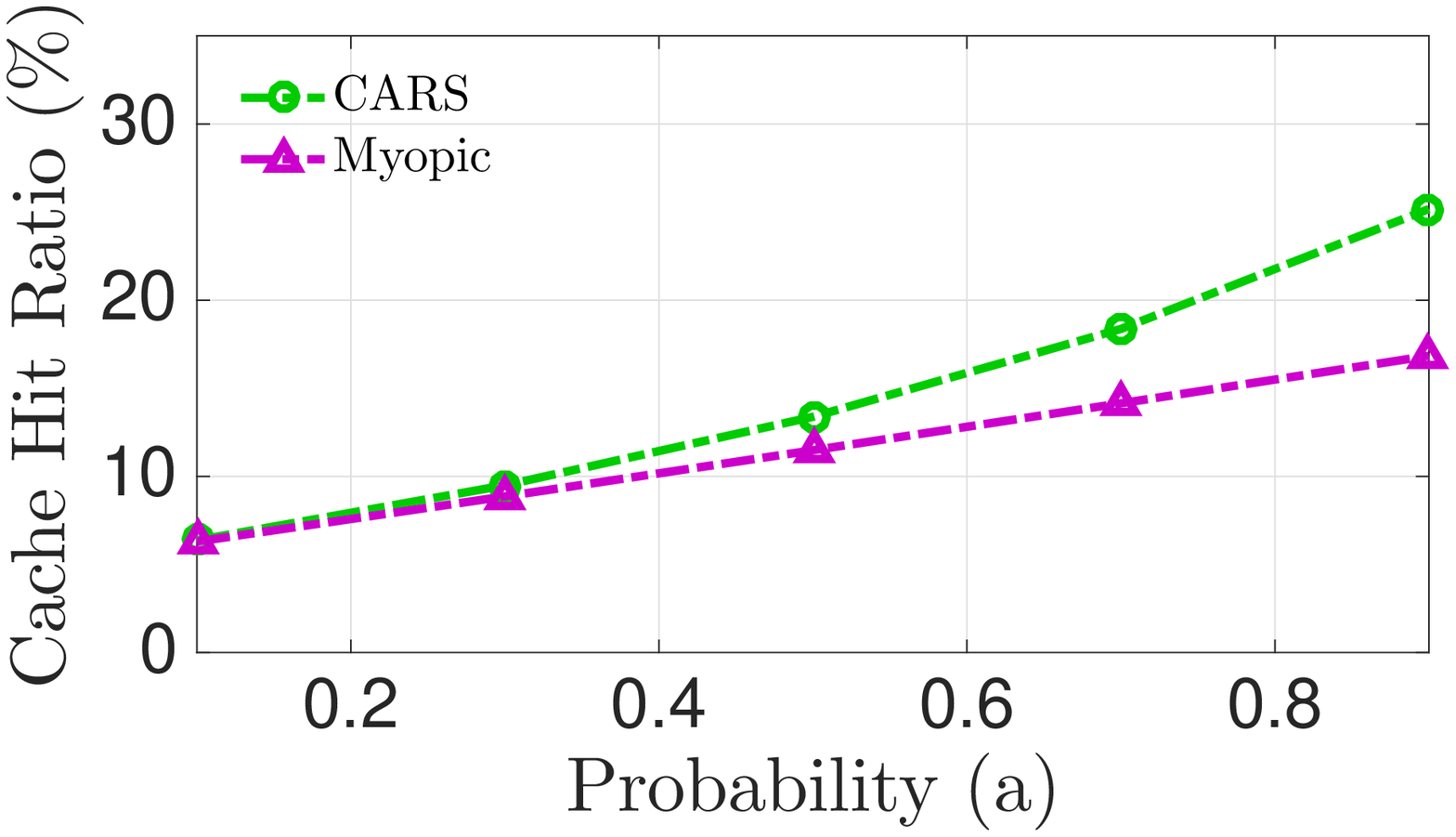}\label{fig:restart}}
\caption{CHR for $N=3$,~(synthetic scenario)\\(a)~$q=85\%,~s=0.2,~C/K=4\%$\\~(b)~$q=90\%,~s=0.6,~C/K=2.5\%$}
\end{figure}



\section{Related Work}
\label{sec:related}


\myitem{Mobile Edge Caching.} Deploying small cells (SCs) over the existing macro-cell networks infrastructure, has been extensively studied and is considered a promising solution that could handle the existing and predicted massive data demands~\cite{hetnets-commag-2012,femtocells-survey-jsac-2012,hetnets-paradigm-shift}. However, this densification of the cellular network will undoubtedly impose heavier load to the backhaul network. Taking advantage of the skewness in traffic demand, it has been suggested that caching popular content at the ``edge'' of the network, at SCs~\cite{femto}, user devices~\cite{sermpezis2014,Hui-offloading,Pavlos-Offload2016}, or vehicles~\cite{Whitbeck-offloading,vigneri2016} can significantly relieve the backhaul. 
Our work proposes an orthogonal and complementary approach for increasing the caching efficiency. We modify the recommendation algorithm to point the users towards the cached content, when this is possible and satisfies the quality of user experience. This can bring further gains in cache hit ratio, on top of existing caching algorithms/architectures.
\oldtext{
However, it is essential to keep in mind that the video streaming options a user has, are drawn out of an immense content catalogue and in a way, it becomes a necessity to reduce this size by manipulating users' requests. Therefore, networking research community is slowly taking steps towards employing already existing means (such as recommendations or related lists) to nudge the user in favor of the cached content.
}


\myitem{Caching and Recommendations.} The interplay between recommendation systems and caching has been only recently considered in literature, e.g., for peer-to-peer networks~\cite{content-recommendation-swarming}, CDNs~\cite{what-should-you-cache-nossdav,cache-centric-video-recommendation}, or mobile/cellular networks~\cite{sermpezis-sch-globecom,chatzieleftheriou2017caching,liu2018learning,sch-chants-2016}. Closer to our study, are the works in~\cite{cache-centric-video-recommendation,chatzieleftheriou2017caching,liu2018learning,sermpezis-sch-globecom} that consider the promotion/recommendation of contents towards maximizing the probability of hitting a local cache. Leveraging the high influence of YouTube recommendations to users, the authors of~\cite{cache-centric-video-recommendation} propose a reordering method for the \textit{related list} of videos; despite its simplicity, this method is shown to improve the efficiency of CDNs. \cite{chatzieleftheriou2017caching} considers the joint problem of caching and recommendations, and proposes a heuristic algorithm that initially places contents in a cache (based on content relations) and then recommends contents to users (based on cached contents). At the selection of the recommendations, \cite{chatzieleftheriou2017caching} considers a single request per user, whereas our work considers a sequential content consumption model, which is closer to user behavior in services such as YouTube, Netflix, Spotify, etc. Similarly to~\cite{chatzieleftheriou2017caching}, in~\cite{liu2018learning}, a single access user is considered. The caching policy in~\cite{liu2018learning} is based on machine learning techniques, the users' \textit{behavior} is estimated through the users' interaction with the recommendations and this knowledge is being exploited at the next BS cache updates. Finally,~\cite{sermpezis-sch-globecom} studies the problem of recommendation-aware caching (in contrast to cache-aware recommendations in this paper). Assuming a content provider/service that offers alternative content recommendation or delivery,~\cite{sermpezis-sch-globecom} proposes near-optimal approximation algorithms for content placement in mobile networks with single-cell and multi-cell (e.g., similarly to~\cite{femto}) user association. 



\section{Conclusions}
\label{sec:conclusions}

In this paper we studied the problem of cache-friendly content recommendations in the context of mobile edge caching. We first introduced a model for sequential content requests over a recommendation system, which captures the user behavior in popular services such as YouTube, Netflix, Spotify, etc. Then, we proposed \textit{CARS}, a cache-aware recommendation algorithm that can increase the caching efficiency (for the network operator), without losing in quality of recommendations (for the user). Our simulation results showed that \textit{CARS} outperforms methods that do not take into account sequential content consumption. A promising future research direction work is to consider the joint problem of caching and recommendations --under sequential requests-- in order to fully exploit the potential of modern communication networks.

\bibliographystyle{ieeetr}


\end{document}